\documentclass{sig-alternate}

\usepackage{amsthm}
\usepackage{amssymb,amsmath}
\usepackage{graphics}
\usepackage{amsfonts}
\usepackage{mathrsfs}
\usepackage{amscd}
\usepackage{comment}
\usepackage{color}
\usepackage{url}
\usepackage[plainpages=false,pdfpagelabels,colorlinks=true,citecolor=blue,hypertexnames=false]{hyperref}

\newcommand{\bZ} { {\mathbb{Z}}}

\def\Dtx{{D_t^x}}
\def\Dtxbar{{D_t^{\bar{x}}}}

\newcommand{\cC} { {\mathcal{C}}}

\newcommand{\cP} { {\mathcal{P}}}

\newcommand{\vx} { {\bf x}}

\newcommand{\frako} { {\mathfrak o}}

\newcommand{\pa}{\partial}

\newcommand{\resid}{\operatorname{residue}}
\newcommand{\resul}{\operatorname{resultant}}

\def\lclm{\operatorname{lclm}}

\usepackage[plainpages=false,pdfpagelabels,colorlinks=true,citecolor=blue,hypertexnames=false]{hyperref}

\newtheorem{theorem}{Theorem}
\newtheorem{prop}[theorem]{Proposition}

\newtheorem{lemma}[theorem]{Lemma}

\newtheorem{problem}[theorem]{Problem}

\newtheorem{example}[theorem]{Example}

\def\qed{\quad\rule{1ex}{1ex}}
\theoremstyle{remark}
\newtheorem*{remark}{Remark}

\begin{document}

\title{Telescopers for Rational and\\ Algebraic Functions via Residues}

\numberofauthors{3}

\author{%
 \alignauthor Shaoshi Chen\titlenote{Supported by the National Science Foundation (NFS) grant CCF-1017217.}\\[\smallskipamount]
      \affaddr{\strut Department of Mathematics}\\
      \affaddr{\strut NCSU}\\
      \affaddr{\strut Raleigh, NC 27695, USA}\\[\smallskipamount]
      \email{\strut schen21@ncsu.edu}
 \alignauthor \strut Manuel Kauers\titlenote{Supported by the Austrian Science Fund (FWF) grant Y464-N18.}\\[\smallskipamount]
      \affaddr{\strut RISC}\\
      \affaddr{\strut Johannes Kepler University}\\
      \affaddr{\strut 4040 Linz, Austria}\\[\smallskipamount]
      \email{\strut mkauers@risc.jku.at}
 \alignauthor \strut Michael F.\ Singer\raisebox{1.4ex}{$\ast$}\\[\smallskipamount]
      \affaddr{\strut Department of Mathematics}\\
      \affaddr{\strut NCSU}\\
      \affaddr{\strut Raleigh, NC 27695, USA}\\[\smallskipamount]
      \email{\strut singer@ncsu.edu}
}

\maketitle
\begin{abstract}
We show that the problem of constructing telescopers for functions
of~$m$ variables is equivalent to the problem of constructing telescopers
for algebraic functions of~$m-1$ variables and present a new algorithm to
construct telescopers for algebraic functions of two variables. These considerations
are based on analyzing the residues of the input. According to experiments, the resulting
algorithm for rational functions of three variables is faster than known algorithms, at least
in some examples of combinatorial interest.  The algorithm for algebraic
functions implies a new bound on the order of the telescopers.
\end{abstract}

\kern-\medskipamount


\kern-\medskipamount


\kern-\medskipamount


\section{Introduction}\label{SECT:intro}

The problem of creative telescoping is to find, for a given ``function'' $f$ in several
variables $t_1,\dots,t_n$, $x_1,\dots,x_m$, linear differential operators $L$ involving
only the $t_i$ and derivations with respect to the~$t_i$, and some
other ``functions'' $g_1,\dots,g_m$ such that
\[
  L(f) = D_{x_1}(g_1) + \cdots + D_{x_m}(g_m),
\]
where~$D_{x_j}$ denotes the derivative with respect to~$x_j$. The main motivation for
computing such operators $L$ (called ``telescopers'' for~$f$) is that, under suitable
technical assumptions on $f$ and the domain~$\Omega$, these operators have the
definite integral
\[
  F(t_1,\dots,t_n)=\int_\Omega f(t_1,\dots,t_n,x_1,\dots,x_m)\, dx_1 \cdots dx_m
\]
as a solution. Once differential operators for $F$ have been found, other
algorithms can next be used for determining possible closed forms, or asymptotic
information, or recurrence equations for the series coefficients of~$F$.

There are general algorithms for computing telescopers when the input~$f$ is
holonomic~\cite{Zeilberger1990,lipshitz88,Wilf1992,takayama92,chyzak2000} as
well as special-purpose algorithms designed for restricted input classes
\cite{Zeilberger1990, Zeilberger1991, bostan10}.  The focus in the present paper
is on two such restricted input classes: rational and algebraic functions of
several variables.  Our first result is that an algorithm for computing
telescopers for rational functions of~{$m$} variables directly leads to an
algorithm for computing telescopers for algebraic functions of~{$m-1$}
variables and {\em vice versa} {(Section~\ref{sec:equiv})}.  Our second
result is a new algorithm for creative telescoping of algebraic functions of two
variables (Section~\ref{SEC:alg}), which, by the equivalence, also implies a new
algorithm for creative telescoping of rational functions of three variables.
The algorithm for algebraic functions is mainly interesting because it implies a
new bound on the order of the telescoper in this case
(Theorem~\ref{thm:algbound}), while the implied algorithm for rational functions
is mainly interesting because {at least for some examples it provides an
  efficient alternative to other methods} (Section~\ref{sec:impl}).

For a precise problem description, let~$k$ be a field of characteristic
zero, and~$k(t, \vx)$ be the field of rational functions in~$t$
and~$\vx=(x_1, \ldots, x_m)$ over~$k$. Let~$\hat{\vx}_m$ denote the~$m-1$ variables~$x_1, \ldots, x_{m-1}$.
The algebraic closure of a field~$K$ will be denoted by~$\overline{K}$.
The usual derivations~$\pa/\pa_t$ and~$\pa/\pa_{x_i}$
are denoted by~$D_t$ and~$D_{x_i}$, respectively.
Let $k(t)\langle D_t \rangle$ be the ring of linear differential
operators in~$t$ with coefficients in~$k(t)$.  Then we are interested in the following two problems:

\begin{problem}\label{PB:trirat}
Given~$f\in k(t, \vx)$, find a nonzero operator~$L\in k(t)\langle D_t \rangle$
such that
\[
    L(f) = D_{x_1}(g_1) + \cdots + D_{x_m}(g_m) \quad \text{for some~$g_j\in k(t, \vx)$.}
\]
Such an~$L$ is called a \emph{telescoper} for~$f$, and the rational functions~$g_1, \dots, g_m$ are
called \emph{certificates} of~$L$.
\end{problem}

\begin{problem}\label{PB:bialg}
Given~$\alpha\in \overline{k(t, \hat{\vx}_m)}$, find a nonzero
operator~$L\in k(t)\langle D_t \rangle$ such that
\[
  L(\alpha) {=} D_{x_1}(\beta_1) + \cdots + D_{x_{m-1}}(\beta_{m-1})\text{ for some~$\beta_j\in
  \overline{k(t, \hat{\vx}_m)}$.}
\]
Such an~$L$ is called a \emph{telescoper} for~$\alpha$, and the algebraic functions~$\beta_1, \ldots, \beta_{m-1}$ 
are called \emph{certificates} of~$L$.
\end{problem}

Both the equivalence of these two problems and the new algorithm for
Problem~\ref{PB:bialg} {(when $m=2$)} are based on the general idea of
eliminating residues in the input.  As an introduction to this approach, let us
consider the problem of finding a telescoper and certificate for a rational
function in two variables, that is, given a \emph{rational} function $f \in
k(t,x)$, we want to find a nonzero $L \in k(t)\langle D_t \rangle$ such that
$L(f) = D_x(g)$ for some $g\in k(t, x)$.  We may consider $f$ as an
element of $\overline{K}(x)$, where $K = k(t)$, and as such we may write
\begin{equation} \label{eqn1}
  f = p + \sum_{i=1}^n \sum_{j=1}^{m_i} \frac{\alpha_{i,j}}{(x-\beta_i)^j},
\end{equation}
where $p\in {K}[x]$, the $\beta_i$ are the roots in $\overline{K}$ of
the denominator of $f$ and the $\alpha_{i,j}$ are in~$\overline{K}$.
We refer to the element $\alpha_{i,1}$ as the \emph{residue} of~$f$ at~$\beta_i$.
Using Hermite reduction, one sees that a
rational function $h \in K(x)$ is of the form $h = D_x(g)$ for some $g \in
K(x)$ if and only if all residues of $h$ are zero.  Therefore to find a
telescoper for $f$ it is enough to find a nonzero operator $L \in K\langle
D_t \rangle$ such that $L(f)$ has only zero residues.  For example assume that
$f$ has only simple poles, i.e., $f = \frac{a}{b}, a,b \in K[x]$, $\deg_xa <
\deg_x b$ and $b$ squarefree.  We then know that the Rothstein-Trager resultant
\cite{Trager1976,Rothstein1977}
\[
  R := \resul_x(a-zD_x(b), b)\in K[z]
\]
is a polynomial whose roots are the residues at the poles of~$f$.  Given a
squarefree polynomial in $K[z]=k(t)[z]$, differentiation with respect to $t$ and
elimination allow one to construct a nonzero linear differential operator $L \in
k(t)\langle D_t \rangle$ such that $L$ annihilates the roots of this polynomial.
Applying $L$ to each term of \eqref{eqn1} one sees that $L(f)$ has zero residues
at each of its poles. Applying Hermite reduction to $L(f)$ allows us to find a
$g$ such that $L(f) = D_x(g)$.

The main idea in the method described above is that nonzero residues are the
obstruction to being the derivative of a rational function and one constructs a
linear operator to remove this obstruction.  Understanding how residues form an
obstruction to integrability and constructing linear operators to remove this
obstruction will be the guiding principal that motivates the results which
follow.

\smallskip

The authors would like to thank Barry Trager for useful discussions and
outlining the proof of Proposition~\ref{der:ramified}.

\section{Telescopers for rational\hskip0pt plus1fill\break functions}\label{sec:equiv}

\subsection{Rational and algebraic integrability}

In this section, we give a criterion which decides whether or not~$1$ is a
telescoper for a rational function in~$k(t, \vx)$.  Again, let~$K=k(t)$.
A rational function~$f\in K(\vx)$ is
said to be \emph{rational integrable} with respect to~$\vx$ if~$f = \sum_{j=1}^m D_{x_j}(g_j)$
for some~$g_j\in K(\vx)$.  An algebraic function~$\alpha\in
\overline{K(\hat{\vx}_m)}$ is said to be \emph{algebraic integrable} with respect to~$\hat{\vx}_m$
if~$\alpha = \sum_{j=1}^{m-1}D_{x_j}(\beta_j)$ for some~$\beta_j\in \overline{K({\hat{\vx}}_m)}$.
By taking traces, one can show that if $\alpha$ is algebraic integrable with respect to~$\hat{\vx}_m$,
then an antiderivative of~$\alpha$ already exists in the field~$K(\hat{\vx}_m)(\alpha)$.

For a rational function~$f\in K(\vx)$, Hermite reduction with respect
to~$x_m$ decomposes~$f$ into
\begin{equation}\label{EQ:hermite}
f = D_{x_m}(r) + \frac{a}{b},
\end{equation}
where $r\in K(\vx)$ and $a, b \in K(\hat{\vx}_m)[x_m]$ such that $\deg_{x_m}(a)<\deg_{x_m}(b)$ and $b$ is
squarefree with respect to~$x_m$. It is clear that $f$ is rational integrable with respect to~$\vx$
if and only if $a/b$ in~\eqref{EQ:hermite} is rational integrable with respect to~$\vx$.
Over the field $\overline{K(\hat{\vx}_m)}$, one can write a rational function $f\in K(\vx)$ as
\[
 f = p + \sum_{i=1}^n \sum_{j=1}^{m_i} \frac{\alpha_{i j}}{(x_m-\beta_i)^j},
\]
where $p\in K(\hat{\vx}_m)[x_m]$ and the~$\alpha_{i j}, \beta_i$ are
in $\overline{K(\hat{\vx}_m)}$. We call $\alpha_{i1}$ the \emph{$x_m$-residue} of~$f$
at~$\beta_i$, denoted by $\resid_{x_m}(f, \beta_i)$.
\begin{prop}\label{PROP:residue}
Let~$f\in K(\vx)$ and~$\beta\in \overline{K(\hat{\vx}_m)}$. Then
\begin{itemize}
\item[(i)] $\resid_{x_m}(f, \beta) =0$ if~$f=D_{x_m}(g)$ for some~$g\in K(\vx)$
\item[(ii)] $D_{x_j}(\resid_{x_m}(f, \beta))=\resid_{x_m}(D_{x_j}(f), \beta)$
for all~$j$ with~$1\leq j\leq m-1$.
\end{itemize}
\end{prop}
\begin{proof}
The first assertion follows by observing the effect of~$D_{x_m}$ on each term in the partial
fraction decomposition of~$g$.
By Hermite reduction, we can decompose~$f$ into
\[
 f = D_{x_m}(r) + \sum_{i=1}^n \frac{\alpha_i}{x_m-\beta_i}.
\]
By the first assertion, either~$\resid_{x_m}(f, \beta) = \alpha_i$ if~$\beta=\beta_i$
or~$\resid_{x_m}(f, \beta)=0$ if~$\beta \neq \beta_i$ for all~$i=1, \ldots, n$.
Applying~$D_{x_j}$ to the two sides of the equation above yields
\begin{align*}
D_{x_j}(f) &= D_{x_j}(D_{x_m}(r))+ \sum_{i=1}^n \left( \frac{D_{x_j}(\alpha_i)}{{x_m}-\beta_i} +
\frac{\alpha_iD_{x_j}(\beta_i)}{({x_m}-\beta_i)^2}\right)\\
&=D_{x_m}\left(D_{x_j}(r) - \sum_{i=1}^n \frac{\alpha_iD_{x_j}(\beta_i)}{{x_m}-\beta_i}\right)
+ \sum_{i=1}^n \frac{D_{x_j}(\alpha_i)}{{x_m}-\beta_i}.
\end{align*}
Then we have either~$\resid_{x_m}(D_{x_j}(f), \beta) = D_{x_j}(\alpha_i)$ if~$\beta=\beta_i$
or~$\resid_{x_m}(D_{x_j}(f), \beta)=0$ if~$\beta \neq \beta_i$ for all~$i=1, \ldots, n$.
The second assertion follows.
\end{proof}
If~$f$ is written as the form in~\eqref{EQ:hermite}, then we have
\[
 \resid_{x_m}(f, \beta_i)=\frac{a}{D_{x_m}(b)}\Big|_{{x_m}=\beta_i} \in K(\hat{\vx}_m)(\beta_i).
\]
Therefore, all the $x_m$-residues of~$f$ are roots of the Rothstein-Trager resultant
(see~\cite{Rothstein1977, Trager1976})
\[
 R := \resul_{x_m}(b, a-zD_{x_m}(b))\in K(\hat{\vx}_m)[z].
\]

\begin{lemma}\label{LEM:ratsol}
Let $f\in K(\vx)$. Then $f$ is rational integrable with respect
to~$\vx$ if and only if all the $x_m$-residues of~$f$ are algebraic integrable with respect
to~$\hat{\vx}_m$.
\end{lemma}
\begin{proof}
By the Hermite reduction and partial fraction decomposition, $f$ can be written as
\[
 f = D_{x_m}(r) + \sum_{i=1}^n \frac{\alpha_i}{x_m-\beta_i},
\]
where~$r\in K(\vx)$, $\alpha_i, \beta_i\in \overline{K(\hat{\vx}_m)}$ and
the~$\beta_i$ are pairwise distinct.

{Suppose that} all the $x_m$-residues~$\alpha_i$ of~$f$ are algebraic integrable with
respect to~$\hat{\vx}_m$, i.e., $\alpha_i=\sum_{j=1}^{m-1} D_{x_j}(\gamma_{i, j})$
for some~$\gamma_{i, j}\in K(\hat{\vx}_m)(\alpha_i)$. Note that for each~$j$ we have
\[
  \frac{D_{x_j}(\gamma_{i,j})}{x_m-\beta_i} = D_{x_j}\left(\frac{\gamma_{i, j}}{x_m-\beta_i}\right)
  + D_{x_m}\left(\frac{\gamma_{i, j}D_{x_j}(\beta_i)}{x_m-\beta_i}\right).
\]
Then we get
\[
  \frac{\alpha_i}{ x_m-\beta_i} =  \sum_{j=1}^{m-1}D_{x_j}\left(\frac{\gamma_{i, j}}{x_m-\beta_i}\right)
  + D_{x_m}\biggl(\sum_{j=1}^{m-1}\frac{\gamma_{i, j}D_{x_j}(\beta_i)}{x_m-\beta_i}\biggr).
\]
Therefore, $f$ is rational integrable with respect to~$\vx$ by taking
\[
  g_j= \sum_{i=1}^n \frac{\gamma_{i,j}}{x_m-\beta_i} \quad \text{and}\quad
  g_m= r + \sum_{i=1}^n \sum_{j=1}^{m-1}\frac{\gamma_{i,j}D_{x_j}(\beta_i)}{x_m-\beta_i}.
\]
Note that all the~$g_j$ and~$g_m$ are in~$K(\vx)$ because~$\gamma_{i,j} \in K(\hat{\vx}_m)(\beta_i)$
and~$\beta_i$ are roots of a polynomial in~$K(\hat{\vx}_m)[x_m]$.

Suppose now that~$f$ is rational integrable with respect to~$\vx$, i.e.,
$f=\sum_{j=1}^{m}D_{x_j}(g_j)$ for some~$g_j\in K(\vx)$.  For any $i\in \{1, 2,
\ldots, n\}$, taking {the $x_m$-residues} of~$f$
and $\sum_{j=1}^{m}D_{x_j}(g_j)$, respectively, and using
Proposition~\ref{PROP:residue} we get
\[
 \resid_{x_m}(f, \beta_i) = \alpha_i = \sum_{j=1}^{m-1}D_{x_j}(\resid_{x_m}(g_j, \beta_i)),
\]
which implies that~$\alpha_i$ is algebraic integrable with respect to~$\hat{\vx}_m$.
\end{proof}

\begin{example}
Let~$f=1/(x_1 + x_2)$. Then the $x_2$-residue of~$f$ at~$-x_1$ is~$1$. Since~$1=D_{x_1}(x_1)$,
$f$ is rational integrable with respect to~$x_1$ and~$x_2$. More precisely,
\[f = D_{x_1}\left(\frac{x_1}{x_1+x_2}\right) + D_{x_2}\left(-\frac{x_1}{x_1+x_2}\right).\]
\end{example}

\begin{example}
Let~$f=1/(x_1x_2)$. Then the $x_2$-residue of~$f$ at~$0$ is~$1/{x_1}$. Since~$1/{x_1}$ has no
antiderivative in~$\overline{K(x_1)}$, $f$ is not rational integrable with respect to~$x_1$ and~$x_2$.
\end{example}

\subsection{Equivalence}

\begin{theorem}\label{THM:equiv}
Let~$f\in k(t, \vx)$.
Then~$L\in k(t)\langle D_t \rangle$ is
a telescoper for~$f$ if and only if~$L$ is a telescoper for every $x_m$-residue of~$f$.
\end{theorem}
\begin{proof}
By a similar calculation as in the proof of Proposition~\ref{PROP:residue}, we have
\begin{equation}\label{EQ:teleres}
L(\resid_{x_m}(f, \beta)) = \resid_{x_m}(L(f), \beta)
\end{equation}
for any~$L\in k(t)\langle D_t \rangle$ and~$\beta\in \overline{k(t, \hat{\vx}_m)}$.
If~$L\in k(t)\langle D_t \rangle$ is a telescoper for~$f$, then~$L(f)=\sum_{j=1}^m D_{x_j}(g_j)$
for some~$g_j\in k(t,\vx)$. By Proposition~\ref{PROP:residue} and Equation~\eqref{EQ:teleres},
for the $x_m$-residue $\alpha:=\resid_{x_m}(f, \beta)$ {at} any pole~$\beta$ of~$f$ with respect to~$x_m$,
we have
\[
 L(\alpha) = \sum_{j=1}^{m-1} D_{x_j}(\resid_{x_m}(g_j, \beta)).
\]
So~$L$ is a telescoper for~$\alpha$. Conversely, assume that~$L$ is a telescoper for
any $x_m$-residue of~$f$. Note that any $x_m$-residue of~$L(f)$ is of the form
$L(\resid_{x_m}(f, \beta))$, which is algebraic integrable by assumption. Then~$L(f)$
is rational integrable by Lemma~\ref{LEM:ratsol}. Therefore, $L$ is a telescoper for~$f$.
\end{proof}

Now we can present an explicit translation
between the two telescoping problems by using Theorem~\ref{THM:equiv}.

If we can solve Problem~\ref{PB:bialg}, then for a rational function~$f\in k(t, \vx)$,
first, we can perform Hermite reduction to decompose~$f$ into~$f=D_{x_m}(r)
+ a/b$; second, we compute the resultant~$R:= \resul_{x_m}(a-zD_{x_m}(b), b)\in k(t, \hat{\vx}_m)[z]$;
finally, we get a telescoper for~$f$ by constructing telescopers for all
the roots of~$R$ in~$\overline{k(t, \hat{\vx}_m)}$ and taking their least common left multiple.

On the other hand, if we can solve Problem~\ref{PB:trirat}, then for an
algebraic function~$\alpha\in \overline{k(t, \hat{\vx}_m)}$ with minimal
polynomial~$F\in k[t,\hat{\vx}_m, x_m]$, we compute a telescoper~$L$ for the rational
function $f=x_mD_{x_m}(F)/F$. Note that $\alpha$ is the $x_m$-residue of~$f$ at~$\alpha$.
Therefore, $L$ is a telescoper for~$\alpha$.

\begin{example}\label{ex:8} Consider the rational function
\[
  f=\frac{2y (1-x) x (x+1) (x+2) (t+x) (x y-y-t^4)}
           {1-x(2-x+(x+1) (x+2) (t+x) (xy-y-t^4)^2)}.
\]
In order to find a telescoper for~$f$, we view $f$ as a rational function in $y$
with coefficients in $k(t,x)$ and determine its residues in $\overline{k(t,x)}$.
Write $a$ and $b$ for the numerator and denominator of~$f$.
Since $b$ is squarefree, the residues $z$ of $f$ are precisely the roots
of the the Rothstein-Trager resultant $\resul_y (a - z D_y(b), b)\in k(t,x)[z]$.
In the present example, these are
\[
  \frac{t^4}{x-1} \pm\frac1{\sqrt{x(x+1)(x+2)(x+t)}}.
\]
According to Theorem~\ref{THM:equiv}, it now suffices to find a telescoper for this algebraic function.
This problem is discussed in the following section.
\end{example}

\section{Telescopers for algebraic\hskip0pt plus1fill\break functions}\label{SEC:alg}

We showed above how focusing on residues can yield a technique to find
telescopers of rational functions by reducing this question to a
similar one for algebraic functions. In this section we describe an algorithm to
solve this latter problem for algebraic functions of two variables. In what follows,
the term ``algebraic function'' will always refer to functions of two variables $t$~and $x$.
When one tries to use residues to  solve the
problem of finding telescopers for algebraic functions one must deal with
several complications.  The first is a technical complication.  One does not
have a global way of expressing a function similar to partial fractions and so
must rely on local expansions.  This forces one to look at {differentials}
rather than functions in order to define the notion of residue in a manner that
is independent of local coordinates.  The second complication is a more
substantial one. There are {differentials} $\alpha dx$ having zero residues
everywhere that are not of the form $d\beta = D_x(\beta)dx$, i.e. $\alpha$ is
not the derivative of an algebraic function. Nonetheless, one knows that there
must exist an operator $L \in k(t)\langle D_t \rangle$ of order equal to twice
the genus of the curve associated to $f$ such that $L(\alpha)dx = d\beta$ for
some algebraic~$\beta$.  This will force us to add an additional step to find
our desired telescoper.  In Section~\ref{sec:derdiff}, we will gather some facts
concerning differentials in function fields of one variable that will be needed
in our algorithm. In Section~\ref{sec:alg} we describe the algorithm.

\subsection{Derivations and Differentials}\label{sec:derdiff}

In this section we review some notation and facts concerning function fields of
one variable ({cf.} \cite{artin,bliss,chevalley51,eichler,manin58}). In the
previous section the results and calculations depended heavily on the notion of
the residue of a rational function of $y$ at an algebraic function $\beta_i$ of
$x$.  In the present section we shall also need to use the notion of a residue
but since we are dealing with algebraic functions instead of rational functions,
the appropriate notion is that of a residue of a differential $\omega$ at a
place $\cP$ of the associated function field~$E$.  We will denote this by
$\resid_\cP\omega$ and refer to the above mentioned books for basic definitions
and properties. We note that when $f \in E = \overline{K(x)}(y)$, and $\beta_i
\in \overline{K(x)}$, then $\resid_y(f,\beta_i) = \resid_\cP\omega$, where $\omega = fdx$ and $\cP$
is the place $(y-\beta_i)$ of~$E$.

Let $K$ be a differential field of charactersitic zero with derivation denoted
by $D_t$ (for example, $K = k(t)$ with $D_t$ as above).  Let $x$ be
transcendental over $K$ and $E = K(x,y)$ an algebraic extension of $K(x)$. We
may extend the derivation $D_t$ to a derivation $\Dtx$ on $K(x)$ by first
letting $\Dtx(x) = 0$ and then taking the unique extension to $E$.  We define a
derivation $D_x$ on $K(x)$ by letting $D_x$ be zero on $K$, $D_x(x) = 1$ and
taking the unique extension of $D_x$ from $K(x)$ to $E$.  We shall also assume
that the constants $E^{D_x} = \{c\in E \mid D_x(c) = 0 \}$ are precisely
$K$. This is equivalent to saying that the minimal polynomial of
$y$ over $K(x)$ is absolutely irreducible ({cf.}~\cite{duval91}).
In \cite{chevalley51}, Chapter~VI, \S7, Chevalley shows that $\Dtx$ can be used
to define a map (which we denote again by $\Dtx$) on differentials such that
$\Dtx(fdx) = (\Dtx(f))dx$. The map $\Dtx$ furthermore has the following
properties:
\begin{enumerate}
\item $\Dtx(dg) = d(\Dtx g)$ for any $g \in E$, and
\item  for any place $\cP$ of $E$ and any differential $\omega$,
  \[
    \resid_\cP(\Dtx\omega) = \Dtx(\resid_\cP(\omega)).
  \]
\end{enumerate}
Given $\alpha \in E$ we will want to find an operator $L \in K\langle
\Dtx\rangle$ and an element $\beta\in E$ such that $L(\alpha) = D_x(\beta)$. In
terms of differentials, this latter equation may be written as $L(\omega) =
d\beta$, where $\omega = \alpha dx$.

We shall have occasion to write our field $E$ as $E = K(\bar{x}, \bar{y})$ for
some other $\bar{x}$ which is transcendental over $K$ and $\bar{y}$ algebraic over
$K(\bar{x})$ and work with the derivation $\Dtxbar$ defined in a similar manner
as above.  We will need to know that if we can find a telescoper with respect to
the derivation $\Dtxbar$ then we can convert this into a telescoper with respect
to~$\Dtx$. The following lemma and proposition allow us to do this.

\begin{lemma}\label{der:lem1}
  Let $x$ and $\bar{x}$ be as above and let $\omega$ be a differential of~$E$.
  For any $i =1, 2,\ldots$ there exists $u_i \in E$ such that
\begin{equation}\label{eqn3}
  (\Dtxbar)^i(\omega) - (\Dtx )^i(\omega)  = du_i.
\end{equation}
\end{lemma}
\begin{proof}
  Write $\omega = \bar{\alpha} d\bar{x}$. Lemma~1 of \cite{manin58} (see also
  Lemma~3 in Chapter~VI, \S7 of \cite{chevalley51}) implies that
  \begin{equation}\label{eqn2}
    \Dtxbar(\omega) - \Dtx(\omega)  = -d(\bar{\alpha} \Dtx(\bar{x})).
  \end{equation}
  Letting $u_1 =  -\bar{\alpha}\Dtx(\bar{x})$, we have equation \eqref{eqn2} for $i = 1$.
  One can verify by induction that \eqref{eqn2} holds for $u_{i+1} = \Dtxbar(u_i) -v_i\Dtx(\bar{x})$,
  where $v_i = (\Dtx )^i[\bar{\alpha} D_{x}\bar{x}]\cdot D_{\bar{x}}(x).$
\end{proof}


\begin{prop}\label{der:prop} Let $\alpha \in E$, $\omega = \alpha dx$,
\[
  (\Dtxbar)^n + a_{n-1}(\Dtxbar)^{n-1} + \ldots + a_0 \in  K\langle \Dtxbar \rangle,
\]
and $\bar{\beta} \in E$ such that
\[
  \bigl((\Dtxbar)^n + a_{n-1}(\Dtxbar)^{n-1} + \ldots + a_0\bigr)(\omega)= d\bar{\beta}.
\]
One can effectively find $\beta \in E$ such that
\[
  \bigl((\Dtx)^n + a_{n-1}(\Dtx)^{n-1} + \ldots + a_0\bigr)(\alpha)= D_x(\beta).
\]
\end{prop}
\begin{proof}  From Lemma~\ref{der:lem1} we have that
\begin{alignat*}1
  &\bigl((\Dtxbar)^n + a_{n-1}(\Dtxbar)^{n-1} + \ldots + a_0\bigr)(\omega) \\
  &{}=((\Dtx)^n(\omega)+du_n)+ a_{n-1}((\Dtx)^{n-1}(\omega)+du_{n-1})\\
  &\qquad{}+ \ldots + a_0\omega.
\end{alignat*}
Therefore, taking into account that the $a_i$ belong to~$K$,
\begin{alignat*}1
 &\bigl((\Dtx)^n  + a_{n-1}(\Dtx)^{n-1}  + \ldots + a_0\bigr)(\omega) \\
 &= d\bigl(\bar{\beta} - u_n - a_{n-1}u_{n-1} - \ldots -a_1u_1\bigr),
\end{alignat*}
which implies the conclusion of the proposition with $\beta = \bar{\beta} - u_n
- a_{n-1}u_{n-1} - \ldots -a_1u_1$.
\end{proof}
In the algorithm described in the next section, we will consider a differential $\omega$ in $E = K(x,y)$ and assume that
\begin{enumerate}
\item $\omega$ has no poles at any place above the place of $K(x)$ at infinity,  and
\item  the places where $\omega$ does have a pole are all unramified above places of $K(x)$.
\end{enumerate}
We describe below an algorithm that allows one to select an $\bar{x} \in E$ such that $E = K(\bar{x}, y)$ and  such that $\omega$ satisfies conditions 1.~and 2.~above with respect to $K(\bar{x})$.  The algorithm of Section~\ref{sec:alg} can be used to produce a telescoper with respect to $\Dtxbar$ and Proposition~\ref{der:prop} allows one to convert this telescoper to a telescoper with respect to $\Dtx$. In the  following proposition, the proof that condition 2.~can be fulfilled was outlined to us by Barry Trager \cite{trager_thesis,trager_personal}.
\begin{prop}\label{der:ramified} Let $\omega$ be a differential in $E = K(x,y)$.  One can effectively find an $\bar{x} \in E$ such that $E = K(\bar{x}, y)$ and
\begin{enumerate}
\item $\omega$ has no poles at any place above the place of $K(\bar{x})$ at infinity,  and
\item  the places where $\omega$ does have a pole are all unramified above places of $K(\bar{x})$.
\end{enumerate}
\end{prop}
\begin{proof} If 1.~does not hold, let $c \in K$ be selected so that $\omega$ has no poles above $x=c$, let
\[\bar{x} = \frac{cx}{x-c}.\]
This change of variables interchanges $c$ and the point at infinity, so 1.~is now satisfied with respect to $K(\bar{x})$ and we shall henceforth abuse notation and assume that 1.~is satisfied with respect to $K(x)$.

Let $\cC$ be a nonsingular curve that is a model of $E$.  The elements of $E$
can be considered as functions on $\cC$. As noted in
\cite[p.~63]{trager_thesis}, ramification occurs when the line of projection
from the curve down to the $x$-axis is tangent to the curve and, for each pole
of $\omega$, there are only a finite number of projection directions that are
tangent to the curve at this pole.  Therefore for all but finitely many choices
of {an} integer~$m$, if we let $\bar{x} = x+my$, $\omega$ will satisfy 2.~with
respect to $K(\bar{x})$.  One can refine this argument and produce a finite set
of integers $m$ that are to be avoided. This is done in the following way.

Let $M$ be an indeterminate and consider the field $E_1 = E(M) = k_1(\bar{x},y)$,
where $k_1 = K(M)$ and $\bar{x} = x+My$. Let $\frako = K[M]$ and assume that
(after a possible change of $y$), $y$ satisfies a monic polynomial over
$\frako[\bar{x}]$.  The behavior of various objects in $E_1$ when one reduces
$\frako$ modulo a prime ideal of $\frako$ is considered in \cite[Chapter~III,
\S6]{eichler}. We shall be interested in reducing modulo ideals of the form
$(M-m)$, where $m$ is an integer.  One can effectively calculate an integral
basis $\{w_i(M)\}$ of the integral closure of $k_1[\bar{x}]$ in $E_1$ ({cf.}
\cite{hoeij_integral,trager_thesis}) and from this a complementary basis
$\{w'_i(M)\}$ (\cite[Chapter~5, \S2]{artin}, \cite[\S22]{bliss}).  In
Chapter~III \S6.2 of \cite{eichler}, Eichler gives a method that will produce a
finite set $S \subset \bZ$ such that for $m \notin S$, the set $\{w_i(m)\}$ is
again an integral basis of the integral closure of $K[\bar{x}]$ in $E$.  This
method can be refined (and the set $S$ slightly increased if need be) so that
$\{w'_i(m)\}$ is also a complementary basis.  Expressing $\omega$ in terms of
this complementary basis,
\[
  \omega = \frac{1}{b(\bar{x})}\sum_{i=1}^n p_i(M,\bar x) w'_i(M) d\bar x,
\]
one sees that $\omega$
will have poles precisely at the zeroes of $b(\bar{x})$.  If one selects $m \in
\bZ$ such that $b(\bar{x})$ is relatively prime to $D(\bar{x})$, the discriminant
of the integral basis $\{w_i(m)\}$, then $\omega$ will not have poles at
ramification points.  The {finitely many} values of $m$ that do not satisfy
this latter condition are roots of
\begin{alignat*}1
  S(M) = \resul_X(&\resul_Y(b(X+MY), F(X,Y)),\\
                  &\resul_Y(D(X+MY), F(X,Y))),
\end{alignat*}
where $F \in K[X,Y]$ is the minimal polynomial of $y$ over $K(x)$.\end{proof}

\subsection{An Algorithm to Calculate Telescopers for Algebraic
  Functions}\label{sec:alg} We assume we are given a function field of one
variable $E = K(x,y)$ and a differential $\omega$ in $E$.  We shall furthermore
assume that $\omega$ satisfies conditions 1.~and 2.~of
Proposition~\ref{der:ramified}.  We will describe an algorithm to find $a_0,
\ldots, a_n \in K$, not all zero, and $\beta \in E$ such that
\[
(a_n(\Dtx)^n + a_{n-1}(\Dtx)^{n-1} + \ldots + a_0)(\omega) = d\beta.
\]
If $\omega = \alpha dx$, then $L = a_n(\Dtx)^n + a_{n-1}(\Dtx)^{n-1} + \ldots +
a_0$ is a telescoper for $\alpha$ with certificate $\beta$.  The algorithm has
two steps.  The first step finds an operator $L_1$ such that applying this
operator to $\omega$ results in a differential $L_1(\omega)$ with only zero residues.  The
second step finds an operator $L_2$ of order at most twice the genus of $E$ and
an element $\beta \in E$ such that $L_2(L_1(\omega)) = d\beta$.

\medskip
\noindent \textbf{Step 1.}  We will describe two methods for constructing an
operator that annihilates the residues of $\omega$.  The first one requires one
to calculate in algebraic extensions of $K$ while the second only requires
calculations in~$K$. Throughout, let $F(x,Y)\in K[x,Y]$ be a minimal polynomial of $y$
over $K(x)$ and let
\[
  \omega=\alpha dx = \frac AB dx
\]
for some $A\in K[x,y]$ with no finite poles and $B\in K[x]$.

\smallskip
\noindent {\it Method 1.}  We make no assumptions concerning ramification at the
poles but for convenience we do assume that the poles of $\omega$ only occur at
finite points. Let $a \in \overline K$ be a root of~$B$. For any branch of
$F(x,Y) = 0$ at $x = a$, we may write
\[
  \omega = p_a(z) dz,
\]
where $z = (x-a)^{1/m}$ for some positive integer $m$ and $p_a$ is a Laurent
series in $z$ with coefficients in $\overline{K}$. One can calculate the
coefficient of $1/z$ in $p_a$ and this will be the residue of $\omega$ at this
place. In this way, one can calculate the possible residues $\{r_1, \ldots,
r_s\}$ of~$\omega$.
Let $K_1$ be  a Galois extension of $K$ containing $\{r_1, \ldots,r_s\}$.
Let $C$ be the field of $D_t$-constants in $K_1$ and
$\{\tilde{r}_1, \ldots, \tilde{r}_\ell\}$ be a $C$-basis of $Cr_1 + \ldots
+Cr_s$. Let $L(Y) = wr(Y, \tilde{r}_1, \ldots \tilde{r}_\ell)$ where
$wr({\ldots})$ is the Wronskian determinant. One sees that $L(Y)$ is a nonzero
linear differential polynomial with coefficients in $K_1$ such that $L_1(r_i) =
0$ for $i=1, \ldots, s$. Define
\[
  L_1(Y) = \lclm\{L^\sigma(Y) \mid \sigma \in G\},
\]
where $G$ is the Galois group of $K_1$ over $K$, $L^\sigma(Y)$ denotes the
linear differential polynomial resulting from applying $\sigma $ to each
coefficient of $L$ and $\lclm$ denotes the least common left multiple.  We then
have that $L_1(Y)$ has coefficients in $K$ and annihilates the residues of~$\omega$.

\smallskip
\noindent {\it Method 2.}
We now assume that $\omega$ has poles only at finite
places and that there is no ramification at the poles. This implies that at any
place corresponding to a pole, we may write $\alpha = \sum_{i\geq i_0}
\alpha_i(x-x_0)^i$ for some $\alpha_i \in \bar{K}$.  Therefore the residue of
$\omega$ at this place is \[\alpha_{-1} = \bigl(D_x[(x-x_0)^{-i_0 -1}\alpha]\bigr)_{x =
  x_0}.\] This is the key to the following, parts of which in a slightly different form
appear in~\cite{BronForm}.

\begin{prop}\label{step1prop} Given $\omega$ as above, one can compute a
  polynomial $R \in K[Z]$ of degree
  \[
    m := \deg_Z(R) \leq \deg_Y(F)\deg_x(B^\ast),
  \]
  with $B^\ast$ the square free part of~$B$,
  such that if $a$ is a nonzero residue of $\omega$ then $R(a) = 0$.
  Furthermore, one can compute a {nonzero operator $L_1 =
  a_m(\Dtx)^m + a_{m-1}(\Dtx)^{m-1} + \ldots + a_0\in K\langle\Dtx\rangle$}
  such that $\tilde{\omega}:= L_1(\omega)$ has residue zero at all places.
\end{prop}
\begin{proof} We may write
\[
 \alpha dx = \frac{A}{B}dx =\frac{A_1}{B_1}dx + \frac{A_2}{B_2^2}dx + \dots + \frac{A_\ell}{B_\ell^\ell}dx, 
\]
where the $A, A_i\in K(x,y)$ are regular at finite places and
$B = B_1B_2^2 \cdots B_\ell^\ell \in K[x]$ is the squarefree decomposition of~$B$.
To achieve our goal it is therefore enough to prove the claim for a differential
of the form $\alpha dx=\frac{A}{B^n}dx$, where $A\in K(x,y)$ is regular at finite places and
$B\in K[x]$ is squarefree. Following \cite{BronForm}, we let $u$ be a
differential indeterminate and let
\[
  h = \frac{(Au^{-n})^{(n-1)}}{(n-1)!} \in K(x,y)\langle u \rangle,
\]
where $K(x,y)\langle u \rangle$ is the ring of differential polynomials in $u$
with coefficients in $K(x,y)$ and $(\dots)^{(i)}$ denotes $i$-fold differentiation
with respect to~$x$. Let $\cP$ be a place where $\alpha$ has a pole and let $a$ and $b$
denote the values of $x$ and $y$ at the place. We note that since $A$ is regular at $\cP$ and  $\cP$ is not ramified, any
derivative of $A$ is also regular at $\cP$ (one needs the hypothesis that
these places are unramified to make this claim).  Taking into account the rules
of differentiation, we see that
\[
 h = \frac{p(x,y,u,u', \ldots , u^{(n-1)})}{q(x)u^t},
\]
where $p(x, Y,z_0, z_1, \ldots , z_{n-1}) \in K[x,Y, z_0, z_1, \ldots , z_{n-1}]$,
$t$ is some positive integer and $q(x)\in K[x]$ does not vanish at $\cP$, i.e. $q(a) \neq 0$.
 Let
 \[\tilde{p} = p(x,Y,B', \tfrac12B'', \tfrac13 B^{(3)}, \ldots , \tfrac1n B^{(n)})\in K[x,Y]\]
 and
 \[
  \tilde{q} = q(x)(B')^t\in K[x].
 \]
 One then shows, as in \cite{BronForm}, that $\tilde{p}(a, b)/\tilde{q}(a)$ is the residue
 of $\frac{A}{B^n} dx$ at~$\cP$.

 The above argument shows that the polynomial
 \begin{alignat*}1
   R = \resul_x\bigl(\resul_Y(\tilde{p}- Z\tilde{q}, F), B\bigr)\in K[Z]
 \end{alignat*}
vanishes at the residues of $\alpha dx$. The degree estimate for $R$
follows from the general degree estimate for resultants which states
for any $S,T\in K[u,v]$ that $\deg_u(\resul_v(S, T))$ is at most
\[
   \deg_u(S)\deg_v(T) + \deg_v(S)\deg_u(T).
\]
This implies first that the inner resultant in the definition of $R$ has
$Z$-degree at most $\deg_Y(F)$. (Note that no degree estimates for $\tilde p$
and $\tilde q$ are needed because $\deg_Z(F)=0$.) Applying the rule again to the
outer resultant gives the desired bound $\deg_Y(F)\deg_x(B)$.

Let $R \in K[Z]$ be the polynomial above.  If necessary, we may replace $R$ by a
squarefree polynomial having the same nonzero roots so we shall assume that $R$
is squarefree and of degree~$m$. Using the fact that $R$ and $\frac{dR}{dZ}$ are
relatively prime, there exist polynomials $R_i \in K[Z]$ of degree at most
$m-1$ such that if $\gamma$ is a root of~$R$, then $D_t^i(\gamma) = R_i(\gamma)$ for $i = 0,
1, \ldots$.  Since each $R_i$ has degree at most $m-1$, there exist $a_m, \ldots
, a_0 \in K$, not all zero, such that $(a_m(\Dtx)^m + a_{m-1}(\Dtx)^{m-1} +
\ldots + a_0)(\gamma) = 0$ for any root $\gamma$ of~$R$.  Using the fact that
$\resid_\cP(\Dtx\omega) = \Dtx(\resid_\cP(\omega))$ for any place $\cP$, one
sees that for $L_1 = a_m(\Dtx)^m + a_{m-1}(\Dtx)^{m-1} + \ldots + a_0$,
$\tilde{\omega} = L_1(\omega)$ has zero residue at any place. \end{proof}

\begin{remark} Although Method~2 does not require calculations in an
algebraic extension of~$K$, one needs the condition on ramification to prove
that it is correct. This condition is painful to verify and although
Propositions~\ref{der:prop} and \ref{der:ramified} imply that we can make a
transformation, if necessary, to guarantee that the differential has poles at
places that are not ramified, making such a transformation can increase the
complexity of the data.  In practice, one should calculate the operator~$L_1$
above without testing if the places at poles are ramified, calculate the
operator $L_2$ as in step~2 below (which requires no assumption concerning
ramification) and then test to see if the resulting operator $L_2\circ L_1$ is a
telescoper by checking if the identity $L_2(L_1(\alpha)) = D_x(\beta)$
holds, a simple calculation in $K(x,y)$. If this equality does not hold, then
one can make a change of variable $\bar{x}:=x+my$ for a random $m$ and try
again.  Proposition~\ref{der:ramified} guarantees that after a finite number of
trials one will succeed.
\end{remark}

\begin{example}[continuing Ex.~\ref{ex:8}] \label{ex:13}
  Let $F=y^2 -x(x+1)(x+2)(x+t)$ and consider
\[
  \omega = \Bigl(\frac{t^4}{x-1}+\frac1y\Bigr)dx=\frac uvdx,
\]
where $u=(x-1)y + t^4x(x+1)(x+2)(t+x)$ and $v=x(x+1)(x+2)(x+t)(x-1)$.
The only pole of $\omega$ is a simple pole at $x=1$, so the residues of $\omega$ are the roots of
\begin{alignat*}1
  &\resid_x(\resid_y(u - z D_x(v), F), v)=({\ldots})(z-t^4)^2z^8,
\end{alignat*}
where $({\ldots})$ stands for some factors which are free of $z$ and therefore irrelevant here.
The only nonzero residue $t^4$ is annihilated by $L_1:=t D_t-4$, so
\[
  \tilde\omega = (tD_t-4){(\omega)} = -\frac{(9t+8x)y}{2x(x+1)(x+2)(x+t)^2}dx
\]
has no nonzero residues.
\end{example}

\begin{remark}
\begin{enumerate}
\item In \cite{trager_thesis}, Trager develops a Hermite
reduction method for algebraic functions which, when applied to the differential
$\omega$ above, shows how one can write $\omega = (D_x(g_1) + g_2) dx$, where
$g_1, g_2 \in E$ and $g_2$ has only simple poles at finite points.  Regretably,
$g_2$ may have poles (of higher order) at infinity. Nonetheless, it would be
interesting to see if Trager's procedure can be used to increase efficiency in
our algorithm.
\item The above argument strongly relies on the fact that we are assuming that the
places where $\omega$ has poles are not ramified above places in $K(x)$.  It
would be of interest to give a method to calculate an operator $L_1$ satisfying
the conclusion of Proposition~\ref{step1prop} without this assumption.
\end{enumerate}
\end{remark}

\medskip
\noindent {\bf Step 2.} Let $\tilde{\omega}$ be as in the conclusion of
Proposition~\ref{step1prop}.  Again using the fact that $\resid_\cP(\Dtx\tilde\omega)
= \Dtx(\resid_\cP(\tilde\omega))$ for any place $\cP$, we have for all $i \in
\bZ$ that $(\Dtx)^i( \tilde{\omega})$ is again a differential with zero residues at
all places.  Such a differential is called a {\em differential of the second
  kind} (\cite{chevalley51}, p.~50) and a differential of the form $d\gamma,
\gamma \in E$ is called an {\em exact differential}. Note that any exact
differential is a differential of the second kind.  Corollary 1 of
(\cite{chevalley51}, p.~130) states that the factor space of the space of
differentials of the first kind by the space of exact differentials is a
$K$-vector space of dimension equal to $2G$, where $G$ is the genus of $E$.
Therefore, there exist $\tilde{a}_{2G}, \ldots , \tilde{a}_0 \in K$, not all
zero, such that for $L_2 = \tilde{a}_{2G}(\Dtx)^{2G} +
\tilde{a}_{2G-1}(\Dtx)^{2G-1} + \ldots + \tilde{a}_0$, $L_2(\tilde{\omega}) =
d\beta$ for some $\beta \in E$. Such $L_2$ and $\beta$ can be found as follows.

Let $\tilde{\omega} = \tilde{\alpha}dx$ and let $[E:K(x)] = m$.  For each $i
\geq 0$, there exist $\alpha_{i,0}, \ldots , \alpha_{i,m-1} \in K(x)$ such that
\[
 (\Dtx)^i(\tilde{\alpha})= ( y, \ldots,y^{m-1})\left(\begin{array}{c}\alpha_{i,0}\\ \vdots
    \\\alpha_{i,m-1}\end{array}\right).
\]
In addition, there exists an $m\times m$ matrix $A$ with entries in $K(x)$ such
that
\[
  (D_x(y), \ldots , D_x(y^{m-1}))= (y, \ldots ,y^{m-1})A.
\]
Let $a_0, \ldots, a_{2G}$ be elements of $K$ and $\beta_0, \ldots, \beta_{m-1}$
elements of $K(x)$. Letting $\beta = \beta_0 + \beta_1y+ \ldots
+\beta_{m-1}y^{m-1}$, the equation
\[
  d\beta = (a_{2G}(\Dtx)^{2G}(\tilde{\alpha})+ \ldots + a_0\tilde{\alpha})dx
\]
is equivalent to
\begin{alignat}1
&  D_x \begin{pmatrix} \beta_0\\ \vdots \\\beta_{m-1}\end{pmatrix}
+ A \begin{pmatrix} \beta_0\\ \vdots \\ \beta_{m-1}\end{pmatrix}\notag\\
&\quad
= a_{2G} \begin{pmatrix}\alpha_{2G,0}\\ \vdots  \\\alpha_{2G,m-1}\end{pmatrix}
+ \ldots
+ a_0 \begin{pmatrix}\alpha_{0,0}\\ \vdots  \\\alpha_{0,m-1}\end{pmatrix}.\label{step2eqn}
\end{alignat}
In \cite{barkatou_rational}, Barkatou describes a decision procedure for
deciding if there exist nontrivial $\beta_0, \ldots, \beta_{m-1} \in K(x)$ and
$a_0, \ldots, a_{2G} \in K$ satisfying \eqref{step2eqn} when $K$ is a computable
field (i.e., the arithmetic operations and derivation are computable and one has
an algorithm to factor polynomials over~$K$). Therefore one can apply this to $K
= k(t)$, where $k$ is a computable field of characteristic
zero to produce a desired $L_2$ and~$\beta$.

\begin{example}[continuing Ex.~\ref{ex:13}]
  Let again $F=y^2-x(x+1)(x+2)(x+t)$ and consider the differential
  \[
    \tilde\omega = -\frac{(9t+8x)y}{2x(x+1)(x+2)(x+t)^2}dx.
  \]
  Since the field $E$ has genus~1 and $\tilde\omega$ has only zero residues,
  there exists a telescoper for $\tilde\omega$ of order~2. Indeed, the algorithm outlined above finds
  that $L_2(\tilde\omega) = d \beta$, where
  \begin{alignat*}1
    L_2&=4 (99 t^5-540 t^4+1055 t^3-870 t^2+256 t) D_t^2\\
       &\quad{}+4 (297 t^4-1269 t^3+1900 t^2-1152 t+256) D_t\\
       &\qquad{}+3 (99 t^3-306 t^2+307 t-96)
  \end{alignat*}
  and
  \[
    \beta = \frac{3 (429 t^3+330 t^2 x-891 t^2-648 t x+384 t+256 x) y}{(t+x)^3}.
  \]
  For the differential $\omega$ from Example~\ref{ex:13}, it follows that we have
  $L\omega=d\beta$ with
  \begin{alignat*}1
    L&=L_2\circ(tD_t-4)=4 (t-2) (t-1) t^2 (99 t^2-243 t+128)D_t^3\kern-5pt\null\\
    &\quad{}+4 t (99 t^4-189 t^3-210 t^2+588 t-256)D_t^2\\
    &\qquad{}-3 (1089 t^4-4770 t^3+7293 t^2-4512 t+1024)D_t\\
    &\qquad\quad{}-12 (99 t^3-306 t^2+307 t-96).
  \end{alignat*}
  By Theorem~\ref{THM:equiv}, this operator $L$ is also a telescoper for the trivariate
  rational function~$f$ from Example~\ref{ex:8}. Certificates $g,h$ with
  \[
    L(f) = D_x(g) + D_y(h)
  \]
  can be obtained from $\beta$ following the calculations in the proof of
  Lemma~\ref{LEM:ratsol}. They are however too long to be printed
  here.
\end{example}

\begin{remark} Telescopers and certificates for {\em holomorphic}
differentials arise in Manin's solution of Mordell's Conjecture
\cite{manin58,manin63} and Step~2 of our procedure is just an effective version
of considerations that appear in these papers. Telescopers for holomorphic
differentials are also referred to as {\it Gauss-Manin Connections}.
\end{remark}

Combining the estimates on the order of the operators computed in steps 1 and~2 gives
the following bound on the order of telescopers for algebraic functions.
It can be viewed as a generalization of Corollary~14 in~\cite{bostan10},
which says that for every rational function $f=A/B\in K(x)$ there exists a telescoper
of order at most $\deg_x B^\ast$, where $B^\ast$ is the square free part of~$B$.

\begin{theorem}\label{thm:algbound}
  Let $E$ be an algebraic extension of $K(x)$, $\alpha=A/B\in E$ so that $A$ is regular
  at finite places and $B\in K[x]$.
  Let $B^\ast$ be the square free part of~$B$. Then there exists
  $\beta\in E$ and a nonzero operator $L\in K\langle D_t\rangle$ with $L(\alpha)=D_x(\beta)$
  and
  \[
    \deg_{D_t}(L) \leq [E:K(x)]\deg_x(B^\ast) + 2\operatorname{genus}(E).
  \]
\end{theorem}

\section{Implementation and other\hskip0ptplus1fill\break examples} \label{sec:impl}

We have produced a prototype implementation of the algorithms described above on
top of Koutschan's Mathematica package
``HolonomicFunctions.m''~\cite{koutschan10} and compared the performance to the
built-in creative telescoping implementations of this package. In order to make
the comparison as fair as possible, we have tried to reuse as much code from
Koutschan's package as possible, so that the timings will not implicitly compare
two different implementations of some subroutine but reflect as closely as
possible the speed-up (or slow-down) offered by the ideas presented above.

Five different methods to solve the creative telescoping problem for a rational
function $f\in k(t,x,y)$ were considered: {\textbf{(CC)}}~first use Chyzak's
algorithm~\cite{chyzak2000} to find a holonomic system $S$ of operators in
$k(t,x)\langle D_t,D_x\rangle$ such that for all $L\in S$ there exists a
rational function $g\in k(t,x,y)$ with $L(f)=D_y(g)$, afterwards apply the same
algorithm to $S$ to obtain a telescoper $L\in k(t)\langle D_t\rangle$ for~$f$;
{\textbf{(CK)}}~first compute $S\subseteq k(t,x)\langle D_t,D_x\rangle$ as in
variant~{(CC)}, then apply Koutschan's ansatz~\cite{koutschan10a} to $S$ to obtain
{a} telescoper $L$ for~$f$; {\textbf{(K)}}~compute a telescoper for $f$ directly
with Koutschan's ansatz; {\textbf{(EC)}}~use the reduction from
Section~\ref{sec:equiv}, then apply Chyzak's algorithm to the resulting
algebraic functions, and then take the least common left multiple of the
results; {\textbf{(EA)}}~use the reduction from Section~\ref{sec:equiv}, then apply
the algorithm from Section~\ref{SEC:alg} to the resulting algebraic functions,
and then take the least common left multiple of the results.

Table~\ref{tab:1} shows the performance of these five approaches for the following examples.

\begin{table*}
  \vspace{-\smallskipamount}
  \begin{center}
  \def\c#1{\hbox to4em{\hss\smash{\raisebox{-1.25ex}{#1}}\hss}}
  \begin{tabular}{r|r|r|r|r|r|r|r|r}
       & \c{CC} & \c{CK} & \c K & \c{EC} & \c{EA} & \multicolumn{3}{|c}{telescoper statistics} \\
       &      &      &      &      &      & order & degree & bytecount \\\hline
     1 & $>$150h & 4000.89 & 469.03 & 1.30 & 1.04 & 3 & 6 & 3464\\ 
     2 & 16029.55 & 40043.01 & $>$100h & 1390.14 & 1646.53 & 6 & 71 & 76472 \\ 
     3 & $>$150h & 350495.88 & $>$150h & 203.44 & 328.08 & 9 & 93 & 140520 \\ 
     4 & 638.70 & 1099.08 & $>$40Gb & 37606.28 & 216201.88 & 10 & 32 & 41840 \\ 
     5 & 23823.70 & 676.13 & 19085.67 & 1114.34 & 3117.43 & 7 & 27 & 25320 
  \end{tabular}
  \end{center}
  \vspace{-\smallskipamount}
  \caption{Runtime comparison for the examples described in the text.}\label{tab:1}
\end{table*}

 \begin{enumerate}
 \item The rational function $f$ from Example~\ref{ex:8} above. This example is not
   representative but was particularly designed to be easy for our algorithms
   and difficult for the known ones.
 \item Here $f:=\tfrac1{xy}h(\frac t{xy},x,y)$ with
   $h(t,x,y) = \bigl(1-\tfrac x{1-x}-\tfrac y{1-y}-\tfrac t{1-t}-\tfrac
   {xy}{1-xy} -\tfrac {xt}{1-xt}-\tfrac{yt}{1-yt}-\tfrac{xyt}{1-xyt}\bigr)^{-1}$.
   This is the problem of enumerating diagonal 3D-Queens walks raised in~\cite{bostan11}.
   Our calculation confirms the correctness of the telescoper conjectured there.
 \item Let now $h(t,x,y)=\bigl(1-\tfrac{xy}{1-xy} -\tfrac {xt}{1-xt}-\tfrac
   {yt}{1-yt}-\tfrac{xyt}{1-xyt}\bigr)^{-1}$ and $f=\frac1{xy}h(\frac t{x^2y},x,y)$.
   This is a variation of the previous problem, with the points $(2n,n,n)$ replacing
   the diagonal and now allowing steps along the axes.
 \item The rational function
   $h(t,x,y) = 2t^2/((1-t)(3-(x+y+t+xy+xt+yt)+3xyt))$
   appears in \cite{pemantle08} as the generating function for the probability of certain
   structures in random groves. See \cite{pemantle08} for details on the combinatorial
   background. Here we compute the diagonal series coefficients of $f$ by applying
   creative telescoping to $f=\frac1{xy}h(\frac t{xy^3},x,y)$. As can be seen in this
   example, our algorithms are in some cases not superior.
 \item {With $h$ as before, we now consider $f=\frac1{xy}h(\frac 1{x^2y^2},x,y)$.
   Note the large difference between CC and CK.}
\end{enumerate}

We have put timings for a number of additional examples on the
website~\cite{www}.  Also our code and the certificates for Example~\ref{ex:13}
can be found there. The examples we tested suggest that the reduction from
rational functions to algebraic functions can cause a decent speed-up. It
does seem to depend on whether the Rothstein-Trager resultant of the input
factors into several small factors or not. If it does, it is advantageous
because solving several small instances of Problem~\ref{PB:bialg} is cheaper
than solving a single big one.  Whether after the reduction, the algorithm of
Section~\ref{SEC:alg} or some other method is applied to the resulting algebraic
functions, makes usually not much of a difference. Our algorithm tends to be faster when
Step~1 in Section~\ref{sec:alg} already finds a great part of the telescoper,
leaving only a small coupled differential system to be solved in Step~2.

 \bibliographystyle{plain}

\newcommand{\SortNoop}[1]{}\def\cprime{$'$}

\end{document}